\documentclass[letterpaper, 10pt, conference]{ieeeconf}      

\IEEEoverridecommandlockouts                
\usepackage{graphics} 
\usepackage{epsfig} 
\usepackage{times} 
\usepackage{amsmath} 
\usepackage{amssymb}  
\usepackage{psfrag}
\usepackage{subfigure}
\usepackage{mathptmx}
\usepackage{xcolor}
\usepackage{enumerate}
\usepackage{tikz}
\usepackage{pgfplots}
\usepackage{paralist}
\usepackage[]{algorithm2e}
\usetikzlibrary{arrows,shapes,trees}
\usepackage{blindtext}
\usepackage[]{algorithm2e}

\usepackage{nicefrac} 
\newtheorem{Theorem}{Theorem} 
\newtheorem{Lemma}{Lemma} 
\newtheorem{Proposition}{Proposition} 
\newtheorem{Remark}{Remark} 

\newcommand{\mc}[1]{\mathcal{#1}}

\usepackage{color}
\usepackage[normalem]{ulem}
\usepackage{soul} 
\soulregister\cite7
\soulregister\ref7
\soulregister\pageref7

\title{\LARGE \bf
Distributed Filter Design for Cooperative $\mathcal{H}_\infty$-type Estimation%
\thanks{This work was supported by the German
Research Foundation (DFG) through
the Cluster of Excellence in Simulation Technology (EXC 310/1) at the
University of Stuttgart and the Australian Research Council under Discovery
Projects funding scheme (Project DP120102152).}}

\author{
Jingbo Wu \and Li Li \and Valery Ugrinovskii \and Frank Allg\"ower
\thanks{ 
J.~Wu and F.~Allg\"ower are with the Institute for Systems Theory and
Automatic Control, University of Stuttgart, Stuttgart, Germany. Email:
{\tt\footnotesize \{jingbo.wu, allgower\}@ist.uni-stuttgart.de}. 
L. Li is with School of Electrical, Mechanical and Mechatronic Systems,
University of Technology, Sydney,  Australia. Email: {\tt\footnotesize
  Li.Li@uts.edu.au}; V.~Ugrinovskii is with the School of Information Technology and Electrical
Engineering, University of New South Wales at the 
Australian Defence Force Academy,
Canberra, Australia, Email: {\tt\footnotesize v.ougrinovski@adfa.edu.au}.
Part of this work was done
during L. Li's visit to the School of Information Technology and Electrical
Engineering, 
University of New South Wales at the Australian Defence Force Academy.}}

\begin{document}

\maketitle

\begin{abstract}
In this paper, we consider the distributed robust filtering problem, where estimator design is based on a set of coupled linear matrix inequalities (LMIs). We separate the problem and show that the method of multipliers can be applied to obtain a solution efficiently and in a decentralized fashion, i.e. all local estimators can compute their filter gains locally, with communications restricted to their neighbors. 
\end{abstract}


\section{Introduction}
Estimator design has been an essential part of controller design ever since
the development of state-space based controllers. A milestone was laid by
the Kalman Filter in 1960 \cite{Kalman1960}. 

While in the classical estimator design one estimator is used for the entire
system, distributed estimators have gained attention since a
distributed Kalman Filter was presented in \cite{OlfatiSaber2007,Carli2008}. In a distributed estimator setup, multiple estimators create an estimate of the system's state, either individually \cite{Ding2012} or cooperatively. In the latter case, even when every single estimator may be
able to obtain an estimate of the state on its own, cooperation reduces the
effects of model and measurement disturbances
\cite{SS-2009}. Also, the situations are not uncommon where individual
estimators are unable to obtain an estimate of the state on their own and
cooperation becomes an essential prerequisite \cite{Ugrinovskii2011a,Ugrinovskii2011}. 
The node estimators may even not have a model 
of the full system, but only know a part of the system \cite{Stankovic2009}.

However, even though the setup consists of distributed estimation units
without a central coordinator, in many known approaches the design process
itself requires a central coordination unit. In some practical application
examples, where the design process can
be done offline, this may not be a significant drawback. On the other hand,
in many applications especially those involving distributed sensor networks
with varying communication topology, a centralized computation of
observer parameters represents a severe limitation. Practicality of a
distributed system demands that the estimator design process is to
be carried out in a distributed manner as well. 
If network needs to adapt to some changes, such as a change in the plant or change in the network structure,
this allows each node to reconfigure using only local communications and
computation only. 

In this paper, we provide a complete analysis of one distributed estimation
problem where such a distributed design scheme is possible. Specifically, 
we adopt the setup from \cite{Ugrinovskii2011a} concerned with the problem of
distributed estimation with $\mathcal{H}_{\infty}$ consensus performance. As a matter
of fact, in \cite{Ugrinovskii2011a} a gradient-descent-type algorithm was
proposed that can be used to calculate the filter gains in a distributed
manner. Although the proposed gradient type algorithm demonstrated a
possibility of computing the estimator parameters in principle, a practical
application of that algorithm is hindered due to slow convergence observed
even in low dimensional examples. Also, implementation of the decentralized
design scheme proposed in  \cite{Ugrinovskii2011a} requires bidirectional
communications between the network nodes, which essentially requires the
communication graph to be undirected for the purpose of the estimator
design. In this paper, we address the problem of designing distributed
estimators by using distributed optimization methods presented in
\cite{Bertsekas1989},\cite{Boyd2010a}. Distributed optimization methods are widely applied in networked systems, see e.g. \cite{Raffard2004},\cite{Rabbat2004},\cite{Necoara2011}. The contribution of this paper is to show that the problem of designing distributed estimators is amenable to the methodology of distributed
optimization as well. Although the design scheme is
proposed for the specific class of algorithms in  \cite{Ugrinovskii2011a},
it illustrates all the steps necessary to devise similar design schemes for
other distributed estimation algorithms and distributed optimization subject to LMI-constraints in general.   

The rest of the paper is organized as follows: We first introduce the notation and some preliminaries on graph theory. Then, we revisit some essential results published in \cite{Ugrinovskii2011a} and discuss there limitations with respect to numerical optimization. Section III is dedicated for introducing the proposed optimization scheme. In Section IV, we give a mathematical example, and Section V concludes the paper.

\section{Preliminaries and Background}

In this section, we introduce the basic definitions and results which our main results will build on.
\subsection{Notation}
Let $P$ be a symmetric matrix. If $P$ is positive definite, it is denoted
$P > 0$, and we write $P < 0$, if $P$ is negative definite. 
$0$ denotes a matrix of suitable dimension, with all entries equal $0$. 
Moreover, for
vectors $x \in \mathcal{R}^n$ we use the Euclidean vector norm $\|x\|=
\sqrt{x^\top x}$ and the weighted vector norm $\|x\|_P= \sqrt{x^\top P x}$
for symmetric matrices $P > 0$. For matrices $A \in \mathcal{R}^{n \times
 m}$, we use the Frobenius norm $\|A\| = \sqrt{tr(A^\top A)}$ and the
induced norm  $\|A\|_2 = \sup_{\|x\|\neq 0}(\|Ax\|/\|x\|)$. 
$\mathcal{L}_2$ denotes the Lebesgue space of $\mathbb{R}^n-$valued
vector-functions $z(\cdot)$, defined on the time interval $[0, \infty)$
with the norm $\|z\|_2 = \sqrt{\int_0^\infty \|z(t)\|^2 dt}$.  The
vectorization $vec(\cdot)$ maps any matrix $A \in \mathcal{R}^{n \times
 m}$ to
the $n \cdot m$-dimensional vector $vec(A)$ formed by the stacked
columns of $A$. 

\subsection{Communication graphs}

In this section we summarize some notation from the graph theory.
We use directed, unweighted graphs $\mc{G} = (\mc{V} , \mc{E})$ to describe the communication topology between the individual agents. 
$\mc{V} = \{v_1,...,v_N\}$ is the set of vertices, where $v_k \in \mc{V}$ represents the $k$-th agent. 
$\mc{E} \subseteq \mc{V} \times \mc{V}$ are the sets of edges, which model the information flow, i.e. the $k$-th agent receives information from agent $j$ if and only if $(v_j,v_k) \in \mc{E}$.
The set of vertices that agent $k$ receives information from is called the neighborhood of agent $k$, which is denoted by $\mathcal{N}_k=\{j: (v_j,v_k) \in \mc{E}\}$. The set of vertices that receive information from agent $k$ is called the out-neighborhood of agent $k$, which is denoted by $\mathcal{M}_k=\{j: (v_k,v_i) \in \mc{E}\}$. The in-degree $p_k$ and out-degree $q_k$ of a vertex $k$ is defined as the number of edges in $\mc{E}$, which have $v_k$ as their head and tail, respectively. 

\subsection{Distributed $\mathcal{H}_{\infty}$ estimation and LMI conditions}
We now present the underlying distributed estimation problem from ~\cite{Ugrinovskii2011a}. It involves
LMI design conditions, which are the main object of interest in the paper. 
Our main
objective is to show that these LMI design conditions are amendable to a distributed solution by using the multiplier method (c.f.~\cite{Bertsekas1989}) . 
 
The distributed estimation problem with $\mathcal{H}_{\infty}$ consensus of estimates
posed in~\cite{Ugrinovskii2011a} involves estimation of the state of the
uncertain LTI system described by the differential equation
\begin{equation}\label{eq:LTI}
\begin{aligned}
\dot{x} &= A x + B \xi(t),
\end{aligned}
\end{equation}
where $x \in \mathbb{R}^{n}$ is the system state variable to be estimated and 
$\xi(t) \in \mathcal{L}_2$ is a disturbance
function. 
$N$ estimators are applied, each estimator receives a $r_k$-dimensional measurement 
\begin{equation}\label{sys:indiv_output}
y_k=C_k x + D_k \xi + \overline{D}_k \eta_k(t).
\end{equation}
In \eqref{sys:indiv_output}, $\eta_k(t) \in \mathcal{L}_2$ represents the measurement uncertainty of the local estimator $k$. In particular, it is assumed that $E_k = D_k  D_k^\top + \overline{D}_k \overline{D}_k^\top > 0$.

\medskip

\begin{Remark}
The assumption that $E_k > 0$ is a standard technical assumption made in
nonsingular $\mathcal{H}_{\infty}$ control problems \cite{Basar-Bernhard}. It is obviously satisfied
in the case when all measurements are affected by disturbances, which is
evidently satisfied in practical applications. This assumption is later
used to guarantee boundedness of the solution set. 
\end{Remark}
\medskip

The estimators 
form a network of
interconnected $\mathcal{H}_{\infty}$ filters of the form
\begin{equation}\label{sys:estimator}
\begin{aligned}
\dot{\hat{x}}_k =& A \hat{x}_k + L_k(y_k-C_k \hat{x}_k) +K_k \sum_{j \in \mathcal{N}_k} (\hat{x}_j- \hat{x}_k)
\end{aligned}
\end{equation}
with initial condition $\hat{x}_k(0) = 0$. Here the matrices $L_k \in
\mathbb{R}^{n \times r_k}$ and $K_k \in \mathbb{R}^{n \times n}$ 
are the filter gains to be designed. 

As it can be seen in \eqref{sys:estimator}, the estimators are distributed,
i.e. the local estimators create an estimation of the system's state $x$,
solely based on the local output $y_k$ and communication with neighbouring
estimators. The problem in ~\cite{Ugrinovskii2011a} was to determine
estimator gains $L_k$, $K_k$ in 
\eqref{sys:estimator} to satisfy natural internal stability and
$\mathcal{H}_{\infty}$ gain conditions. To introduce these conditions, define
the local estimator error as
$e_k= x -\hat{x}_k$, and the estimator disagreement function is defined as
\begin{equation}
\Psi(\hat{x})=\frac{1}{N} \sum_{k=1}^N \sum_{j \in \mathcal{N}_k} \| \hat{x}_j - \hat{x}_k\|^2,
\end{equation}
where $\hat{x} = [\hat{x}_1^\top , ..., \hat{x}_N^\top]^\top$ and $e =
[e_1^\top , ..., e_N^\top]^\top$. 
The estimator design problem 
is concerned with achieving the following properties: 

\begin{compactenum}[(i)]
\item In the absence of model and measurement disturbances (i.e., when
  $\xi$, $\eta_k=0$), the estimation errors decay so that $e_k \to 0$
  asymptotically for all $k=1,...,N$. 

\item The estimators \eqref{sys:estimator} provide guaranteed
  $\mathcal{H}_{\infty}$ performance
\begin{equation}\label{eq:hinf-performance1}
\begin{aligned}
&\sup_{x_0, (\xi, \eta_k) \neq 0} \frac{\int_0^\infty \Psi(\hat{x}(t))dt}{\|x_0\|_P^2 + (1/N)\sum_{k=1}^N \| \eta_k \|_2^2 + \|\xi \|_2^2} \leq \gamma\\
&\frac{1}{N}\sum_{k=1}^N \|e_k \|_2^2 \leq \overline\gamma \left( \|x_0\|^2_P + \frac{1}{N} \sum_{k=1}^N \| \eta_k \|^2_2 + \| \xi \|^2_2 \right),
\end{aligned}
\end{equation}
for some positive definite matrix $P$, some $\overline\gamma > 0$, and performance index $\gamma > 0$.
\end{compactenum}

\medskip

Property (ii) requires both the local estimation errors and the estimator
disagreement to be bounded with respect to the disturbances in an
$\mathcal{H}_\infty-$sense. As shown in \cite{Ugrinovskii2011a},
LMI-conditions can be found, where the solution delivers estimator gains
sufficient for solving the above problem.  To present these LMI conditions, define the matrices
\begin{align*}
\widetilde{A}_k =& A + \alpha_k I - B D_k^\top E_k^{-1} C_k, \\
Q_k =& X_k \widetilde{A}_k + \widetilde{A}_k^\top X_k - C_k^\top E_k^{-1} C_k + \beta(p_k+q_k)I, \\
\widetilde{B}_k =& [B(I-D_k^\top E_k^{-1} D_k) \quad -BD_k^\top E_k^{-1} \overline{D_k}],
\end{align*}
where $X_k \in \mathbb{R}^{n \times n}$ is a symmetric,
positive definite matrix and $\alpha_k, \beta$ are positive 
parameters. 
For the remainder of this paper, we will make two assumptions on the system class.

\noindent\textbf{Assumption 1}
The communication graph $\mathcal{G}$ is connected and balanced, i.e. $q_k = p_k$ for all $k=1,...,N$.

\noindent\textbf{Assumption 2}
For all $k=1,...,N$, the tuple $(\widetilde{A}_k, \widetilde{B}_k)$ is controllable.

\medskip

The LMIs used for designing the estimator gains are proposed as
\begin{align}
\begin{bmatrix}
Q_k-p_k F_k - p_k F_k^\top  &  X_k \widetilde{B}_k \\
*  & -I
\end{bmatrix} & < 0 \label{LMI:stability} \\
\begin{bmatrix}
-\frac{2 \alpha_k}{q_k +1} X_k & -\beta I + F_k & \hdots & -\beta I + F_k \\
* & -\frac{2 \alpha_{j_1^k}}{q_{j_1^k} +1} X_{j_1^k} & \hdots & 0 \\
\vdots & \vdots & \ddots & \vdots \\
* & 0 & \hdots & -\frac{2 \alpha_{j_{p_k}^k}}{q_{j_{p_k}^k} +1} X_{j_{p_k}^k}
\end{bmatrix} & < 0 \label{LMI:couplings}, \\
\begin{bmatrix}\label{LMI:rest_feasible_set}
- \rho X_k & -F_k^\top \\
-F_k & -X_k
\end{bmatrix} & < 0,
\end{align}
where $X_k, F_k$, and $\beta$ are the variables, $\rho>0$ is a constant
parameter, and $\mathcal{N}_k=\{j_1^k,...,j_{p_k}^k\}$. We can now formulate a variation of the main result from
\cite{Ugrinovskii2011a}. 

\medskip

\begin{Proposition}
Suppose the interconnection graph $\mathcal{G}$ and the parameters
$\alpha_k > 0,~ k = 1, . . . , N$, are such that the set
\begin{equation}
\Gamma = \{\beta > 0: \text{\eqref{LMI:stability}-\eqref{LMI:rest_feasible_set}  are feasible for } k=1,...,N\}
\end{equation}
is not empty. For any $\beta \in \Gamma$, one solution to the distributed
estimation problem under consideration, with $\gamma = \frac{1}{\beta}$, is given by the
network of estimators \eqref{sys:estimator} in which
\begin{equation}
\begin{aligned} \label{eq:filter_gains}
K_k &= X_k^{-1} F_k \quad \text{and} \quad 
L_k &= (X_k^{-1} C^\top_k + BD^\top_k)E_k^{-1},
\end{aligned}
\end{equation}
where $X_k$ and $F_k, k = 1, . . . , N$, belong to the feasibility set
of \eqref{LMI:stability} - \eqref{LMI:rest_feasible_set}, corresponding to this particular value of $\beta$. The
weighting matrix $P$ in \eqref{eq:hinf-performance1} is given by $P = (1/N) \sum_{k=1}^N X_k$.
\end{Proposition}

\medskip

\begin{Remark}
Assumption 1 is a restriction toward the class of communication graphs, which is made in order to ensure that the well-known average consensus algorithm is applicable. 
Assumption 2 is used to ensure boundedness of the feasible sets. It is not restrictive, as it represents the worst case of disturbance, and if not satisfied, small hypothetical disturbances can be added to the system description, i.e. additional columns to $B, D_k$ and $\overline{D}_k$.
Furthermore, note that the tuple $(\widetilde{A}_k,C_k)$ is not required to be detectable.
\end{Remark}

\medskip

Since the LMIs \eqref{LMI:stability}, \eqref{LMI:couplings}, \eqref{LMI:rest_feasible_set} are coupled, they may be solved in a centralized manner as the optimization problem
\begin{equation}
\begin{aligned}\label{Prob:centr_optimization}
\min & \;\;( -\beta) \\
\text{subject to } & \eqref{LMI:stability},\eqref{LMI:couplings}, \eqref{LMI:rest_feasible_set}, k=1,...,N,
\end{aligned}
\end{equation}
where the resulting matrices $X_k, F_k$ deliver the estimator gains $L_k, K_k$ according to \eqref{eq:filter_gains}.
In the next section we will explore the separation of the problem and parallel computation in order to solve the problem in a distributed manner.

\section{Distributed Calculation of Filter Gains}

Parallel and distributed computation is thoroughly discussed e.g. in
\cite{Bertsekas1989}, and in this section, we use some of the methods
presented in Section 3 in \cite{Bertsekas1989} to calculate our estimator gains in a distributed fashion.
Solving the optimization problem \eqref{Prob:centr_optimization} can be formulated as a separable problem by defining local representations of the solution variables, $X_j^k$, and $\beta^k$ for all $k=1,\ldots,N$ and $j=k, j_1^k,...,j_{p_k}^k$. 
The tuple of local variables is denoted by 
\begin{equation}
Y_k = (F_k, \beta^k, X_k^k, X_{j_1^k}^k,...,X_{j_{p_k}^k}^k ),
\end{equation}
where the upper index $k$ denotes the representation of a variable used by estimator $k$ and all $X_j^k$ are symmetric, positive definite matrices, and $\beta^k \geq 0$.

\medskip

\noindent\textbf{Problem 1:} Find an iterative algorithm, which creates a sequence $Y_k(t), t \in \mathbb{N}$, such that local representations of the variables converge in the sense that
\begin{equation} \label{eq:beta-convergence}
\begin{aligned}
\lim_{t \to \infty} \left( \beta^{k_1}(t) - \beta^{k_2}(t) \right) &= 0, \\
\end{aligned}
\end{equation}
for all $k_1, k_2 = 1,...,N$ and
\begin{equation} \label{eq:matrix-convergence}
\begin{aligned}
\lim_{t \to \infty} \left( X_j^{k_1}(t)-X_j^{k_2}(t) \right) &= 0 , \\
\end{aligned}
\end{equation}
for all $j=1,...,N$ and $k_1, k_2 \in \mathcal{M}_j \cup j$. All iterations $Y_k(t)$ shall satisfy the LMIs \eqref{LMI:stability}-\eqref{LMI:rest_feasible_set} when setting $\beta=\beta^k, X_k=X_k^k, X_{j_1^k}=X_{j_1^k}^k,...,X_{j_{p_k}^k}=X_{j_{p_k}^k}^k$.
Furthermore, the iteration steps of the local variables $Y_k(t+1)$ shall be calculated in a distributed fashion, i.e. interaction with the neighbors $j \in \mathcal{N}_k$ only.

\medskip

As a first step, in order to ensure that both \eqref{Prob:centr_optimization} and Problem 1 are well-posed, we establish a statement about the boundedness of the feasible set of the LMIs \eqref{LMI:stability}-\eqref{LMI:rest_feasible_set}. The proof of this theorem will later be used in order to ensure that solutions of local optimizations are always attainable.

\medskip

\begin{Theorem}
Suppose the pairs $(\widetilde A_k,\widetilde B_k)$ are controllable. Then, for any $\rho>0$, the
feasible set  
\begin{equation}
\begin{aligned}
\Omega = \{ &(\beta, X_k, F_k, k=1,...,N) | \\
&\eqref{LMI:stability}, \eqref{LMI:couplings}, \eqref{LMI:rest_feasible_set} \text{ hold for } k=1,...,N\}
\end{aligned}
\end{equation}
is bounded.

\begin{proof}
Suppose $(\beta, X_k, F_k, k=1,\ldots,N)\in\Omega$. Using the Schur
complement, it follows from
(\ref{LMI:stability}), (\ref{LMI:rest_feasible_set}) that for an arbitrary
$\tau_k>0$, 
\begin{eqnarray}
  \label{Ric.ineq}
  X_k \widetilde{A}_k + \widetilde{A}_k^\top X_k - C_k^\top E_k^{-1} C_k +
  \beta(p_k+q_k)I && \nonumber \\
  -p_k F_k - p_k F_k^\top  +  X_k \widetilde{B}_k \widetilde{B}_k^\top
  X_k && \nonumber \\
  + \tau_k (F_k^\top X_k^{-1} F_k-\rho X_k)&&<0.  
\end{eqnarray}
Completing the squares on the left-hand side yields
\begin{eqnarray}
  \label{Ric.ineq.1}
  X_k \widetilde{A}_k + \widetilde{A}_k^\top X_k - C_k^\top E_k^{-1} C_k +
  \beta(p_k+q_k)I && \nonumber \\
  +\tau_k(F_k - \frac{p_k}{\tau_k}X_k)^\top X_k^{-1}(F_k - \frac{p_k}{\tau_k}X_k)   && \nonumber \\
  -\tau_k(\frac{p_k^2}{\tau_k^2}+\rho) X_k+  X_k \widetilde{B}_k \widetilde{B}_k^\top
  X_k &&<0.  
\end{eqnarray}
Hence, we conclude that $(\beta,X_k)$ satisfy the Riccati inequality
\begin{eqnarray}
  \label{Ric.ineq.2}
  X_k (\widetilde{A}_k-\frac{p_k^2+\tau_k^2\rho}{2\tau_k}I) +
  (\widetilde{A}_k-\frac{p_k^2+\tau_k^2\rho}{2\tau_k}I)^\top X_k  && \nonumber
  \\
- C_k^\top E_k^{-1} C_k +  \beta(p_k+q_k)I
 +  X_k \widetilde{B}_k \widetilde{B}_k^\top  X_k &&<0.  
\end{eqnarray}
After pre- and post-multiplying (\ref{Ric.ineq.2}) by $X_k^{-1}$,
(\ref{Ric.ineq.2}) reduces to
\begin{eqnarray}
  \label{Ric.ineq.3}
  (\widetilde{A}_k-\frac{p_k^2+\tau_k^2\rho}{2\tau_k}I)X_k^{-1} +
  X_k^{-1}(\widetilde{A}_k-\frac{p_k^2+\tau_k^2\rho}{2\tau_k}I)^\top  && \nonumber
  \\
- X_k^{-1}(C_k^\top E_k^{-1} C_k -  \beta(p_k+q_k)I)X_k^{-1}
 +  \widetilde{B}_k \widetilde{B}_k^\top &&<0.  
\end{eqnarray}

Associated with this Riccati inequality, consider the Riccati equation
\begin{eqnarray}
  \label{Ric.eq}
  (\widetilde{A}_k-\frac{p_k^2+\tau_k^2\rho}{2\tau_k}I)Z_k +
  Z_k(\widetilde{A}_k-\frac{p_k^2+\tau_k^2\rho}{2\tau_k}I)^\top  && \nonumber
  \\
- Z_k(C_k^\top E_k^{-1} C_k - \frac{1}{\gamma}(p_k+q_k)I)Z_k
 +  \widetilde{B}_k \widetilde{B}_k^\top &&=0  
\end{eqnarray}
and define
\begin{eqnarray}
  \label{gamma.def}
  \gamma^\circ=\inf\left\{\begin{array}{cl}\gamma>0\colon & \mbox{equation~(\ref{Ric.eq}) has
      a nonnegative-}\\ & \mbox{definite solution}
    \end{array}
  \right\}.
\end{eqnarray}
From the $\mathcal{H}_{\infty}$ control theory~\cite[Theorems 4.8 and
9.7]{Basar-Bernhard}, it is known that 
the set whose infimum determines $\gamma^\circ$ is nonempty if the pair 
$(\widetilde{A}_k-\frac{p_k^2+\tau_k^2\rho}{2\tau_k}I, C_k)$ is detectable and
the pair $(\widetilde{A}_k-\frac{p_k^2+\tau_k^2\rho}{2\tau_k}I, \tilde B_k)$
is stabilizable. Note that by the condition of the theorem, the pair
$(\widetilde{A}_k, \tilde B_k)$ is controllable; this implies the stabilizability 
of $(\widetilde{A}_k-\frac{p_k^2+\tau_k^2\rho}{2\tau_k}I, \tilde B_k)$. Now, let
us choose $\tau_k>0$ such that all unstable unobservable modes of the
matrix pair  $(\widetilde{A}_k,C_k)$ lie in the region
$\mathrm{Re}\;s<\frac{p_k^2+\tau_k^2\rho}{2\tau_k}$. This will guarantee that     
the pair $(\tilde A-\frac{p_k^2+\tau_k^2\rho}{2\tau_k}I,C_k)$ is
detectable. Thus, we conclude that $\gamma^\circ<\infty$.

The feasibility of the Riccati inequality (\ref{Ric.ineq.3}) also
implies that the following state-feedback $\mathcal{H}_{\infty}$ control problem
involving the system 
\begin{eqnarray}
  \label{VU.sys}
  && \dot x=(\widetilde{A}_k-\frac{p_k^2+\tau_k^2\rho}{2\tau_k}I)^\top x
+C_k^\top u +(p_k+q_k)^{1/2} w, \\
  && z_k= \left[\begin{array}{c} \widetilde{B}_k^\top \\ 0
    \end{array}\right]+\left[\begin{array}{c} 0\\ E_k^{-1/2}
    \end{array}\right]u \nonumber
  \end{eqnarray}
and the $\mathcal{H}_{\infty}$ performance criterion
\begin{eqnarray}
  \label{VU.Hinf}
  \int_0^\infty \|z_k\|^2dt< \frac{1}{\beta}\int_0^\infty \|w\|^2dt \quad
  \forall w\in L_2, \quad (x(0)=0),
\end{eqnarray}
has a solution. Indeed, it follows from (\ref{Ric.ineq.3}) that
\begin{eqnarray}
  \label{Ric.ineq.4}
\lefteqn{  (\widetilde{A}_k-\frac{p_k^2+\tau_k^2\rho}{2\tau_k}I-X_k^{-1} C_k^\top
  E_k^{-1} C_k)X_k^{-1}} && \nonumber \\ && +
  X_k^{-1}(\widetilde{A}_k-\frac{p_k^2+\tau_k^2\rho}{2\tau_k}I-X_k^{-1} C_k^\top E_k^{-1} C_k)^\top  \nonumber
  \\
&& + X_k^{-1}(C_k^\top E_k^{-1} C_k +  \beta(p_k+q_k)I)X_k^{-1}
 +  \widetilde{B}_k \widetilde{B}_k^\top <0.  
\end{eqnarray}
Since $X_k^{-1}>0$ and (\ref{Ric.ineq.4}) is a strict inequality,
the matrix  
\begin{equation}
\label{Hurw}
\tilde
A-\frac{p_k^2+\tau_k^2\rho}{2\tau_k}I-C_k^\top E_k^{-1} C_kX_k^{-1}
\end{equation}
is Hurwitz. Thus, the
closed loop system consisting of the system (\ref{VU.sys}) with $w=0$ and the
state-feedback controller 
\[
u=-E_k^{-1} C_kX_k^{-1}x
\]
is exponentially stable. Also, 
using the completion of squares, it is easy to show from  (\ref{Ric.ineq.3})
that the above controller guarantees the $\mathcal{H}_{\infty}$ attenuation property
(\ref{VU.Hinf}). Since the pair $(\widetilde{A}_k,\widetilde{B}_k)$ is
controllable, these observations guarantee that the Riccati equation
(\ref{Ric.eq}) with $\gamma=\frac{1}{\beta}$ has a unique nonnegative
definite stabilizing solution $Z_k$ (e.g., see
\cite[Theorem~3.2.2]{PUSB}). Thus,
$\frac{1}{\beta}>\gamma^\circ$. Furthermore, since
$(\widetilde{A}_k,\widetilde{B}_k)$ is assumed to be controllable, $Z_k>0$
and is invertible. 

From Theorem 4.8 in \cite{Basar-Bernhard}, we know that $\gamma^\circ>0$.
These observations imply that $\beta<{\gamma^\circ}^{-1}$. Also, using the
relationship between solutions to the Riccati equation (\ref{Ric.eq}) and
the corresponding Riccati inequality
(\ref{Ric.ineq.3})~\cite[Lemma~8.1]{GA-1994}, it follows that
$X_k<Z_k^{-1}$. 

This discussion leads us to conclude that there exist upper bounds on
feasible $\beta$ and $\|X_k\|$. Indeed, $\gamma^\circ$ and $Z_k$ are defined
using the conditions involving the properties of the matrices
$\widetilde{A}_k$, $C_k$ and $B_k$ and the constants $\rho$, $p_k$. Hence,
these constant and the matrix are not dependent on the choice of the
feasible $\beta$ and $X_k$. 

It remains to show that there is an upper bound on the feasible $F_k$ as
well. Using the Schur complement, (\ref{LMI:rest_feasible_set}) is equivalent to 
$  F_k^\top X_k^{-1} F_k<\rho X_k$.
This further implies
\begin{align*}
{X_k^{-1/2}} F_k^\top X_k^{-1} F_k {X_k^{-1/2}} &< \rho I \\
tr \left({X_k^{-1/2}} F_k^\top X_k^{-1} F_k {X_k^{-1/2}} \right) &< n \rho \\
\| {X_k^{-1/2}} F_k {X_k^{-1/2}} \| &< \sqrt{n\rho}.
\end{align*}
For the Frobenius-norm of $F_k$, we can now conclude
\begin{align*}
\|F_k \| =& \| {X_k^{1/2}} {X_k^{-1/2}} F_k {X_k^{-1/2}} {X_k^{1/2}} \|  \\
\leq& \| {X_k^{1/2}} \| \| {X_k^{-1/2}} F_k {X_k^{-1/2}}\|  \| {X_k^{1/2}} \| \\
<& \sqrt{n\rho} \| {X_k^{1/2}} \|^2,
\end{align*}
which is bounded due to boundedness of $X_k$.
\end{proof}

\end{Theorem}

\medskip

Now, 
the decoupled version of the LMI conditions \eqref{LMI:stability},
\eqref{LMI:couplings}, \eqref{LMI:rest_feasible_set} is proposed as
\begin{align}
\begin{bmatrix}
Q_k^k-p_k F_k - p_k F_k^\top  &  X_k^k \widetilde{B}_k \\
*  & -I
\end{bmatrix}  \leq  - &\begin{bmatrix}
\delta X_k^k & 0 \\ 0 & 0
\end{bmatrix} \label{LMI:stability_sep} \\
\hspace{-0.5cm}\begin{bmatrix}
-\frac{2 \alpha_k}{q_k +1} X_k^k & -\beta^k I + F_k & \hdots & -\beta^k I + F_k \\
* & -\frac{2 \alpha_{j_1^k}}{q_{j_1^k} +1} X_{j_1^k}^k & \hdots & 0 \\
\vdots & \vdots & \ddots & \vdots \\
* & 0 & \hdots & -\frac{2 \alpha_{j_{p_k}}}{q_{j_{p_k}} +1} X_{j_{p_k}}^k
\end{bmatrix} & \leq -\delta \; I \label{LMI:couplings_sep} \\
\begin{bmatrix}
- \rho X_k^k & -F_k^\top \\
-F_k & -X_k^k
\end{bmatrix} & \leq 0, \label{LMI:rest_feasible_set_sep}
\end{align}
with
$ Q_k^k = X_k^k \widetilde{A}_k + \widetilde{A}_k^\top X_k^k - C_k^\top E_k^{-1} C_k + \beta^k(p_k+q_k)I$.
Note that the LMI-conditions are formulated as non-strict inequalities, but with additional parameter $\delta>0$. 
However, as $\delta$ can be chosen arbitrarily small, it introduces no conservativeness. 


We denote the feasible set of the $k$-th group of the LMIs as 
$\Omega_k = \{Y_k \;\; | \;\; \eqref{LMI:stability_sep}, \eqref{LMI:couplings_sep}, \eqref{LMI:rest_feasible_set_sep} \text{ hold true}\}$.
Then, the separable convex program can be written as 
\begin{equation}\label{prob:separable_optimization}
\begin{aligned}
 \text{minimize} & (-\sum_{k=1}^N \beta^k) \\
 \text{subject to }& Y_k \in \Omega_k, \quad
  \beta^k = \widetilde{\beta} \\
 & X_k^k = \widetilde{X}_k, \quad
  X_{j_1^k}^k = \widetilde{X}_{j_1^k}, \quad
 ... \quad,
  X_{j_{p_k}^k}^k = \widetilde{X}_{j_{p_k}^k}
\end{aligned}
\end{equation}
for every $k=1,...,N$. Here, $\widetilde{\beta}, \widetilde{X}_j, j=1,...,N$  are additional variables that are needed to make the problem separable.

\medskip

\begin{Remark}
The optimization problem \eqref{prob:separable_optimization} can be varied in the way that for a given performance parameter $\beta >0$, filter gains for \eqref{sys:estimator} are to be found. Then, \eqref{prob:separable_optimization} turns to a pure feasibility problem without optimization objective, and therefore, the variables $\beta^k, \widetilde \beta$ and their iterations in the following algorithm can be omitted.
\end{Remark}

\medskip

The dual problem has the form
\begin{equation}
\begin{aligned}\label{prob:dual_optimization}
&\text{maximize } \textbf{q}(\widetilde{\Lambda}^1,...,\widetilde{\Lambda}^N) \\
\end{aligned}
\end{equation}
where $\widetilde{\Lambda}^k=( \lambda^k, \Lambda^k_k, \Lambda^k_{j_1^k},..., \Lambda^k_{j_{p_k}^k})$ for $k=1,...,N$
is the suitable tuple of Lagrange multipliers and the dual function $\textbf{q}(\cdot)$ is defined as
\begin{equation}
\textbf{q}(\widetilde{\Lambda}^1,...,\widetilde{\Lambda}^N) =  \inf_{Y_k \in \Omega_k, k=1,...,N} L(Y_1,...,Y_N, \widetilde{\Lambda}^1,...,\widetilde{\Lambda}^N).
\end{equation}
$L(\cdot)$ is the augmented Lagrangian function (cf. \cite{Bertsekas1989})
\begin{equation}\label{eq:Lagrangian_total}
\begin{aligned}
L(Y_k, \widetilde{\Lambda}^k) =& \sum_{k=1}^N \left( -\beta^k + \lambda^k (\widetilde\beta-\beta^k) + \frac{c}{2} | \widetilde\beta-\beta^k |^2 \right) \\
 &+ \sum_{k=1}^N\sum_{j \in \mathcal{N}_k \cup k} \!\!\! \left( tr\left(\Lambda^{k \top}_{j} (\widetilde X_j - X_j^k) \right) + \frac{c}{2}\| \widetilde X_j - X_j^k \|^2 \right)
\end{aligned}
\end{equation}
with design parameter $c>0$.
The optimization problem \eqref{prob:dual_optimization} can now be solved iteratively
with Algorithm \ref{alg:iteration}, which is initialized with $Y_k(0) \in \Omega_k$, $\lambda>0$ and symmetric $\Lambda_j^k>0$.

\begin{algorithm}\label{alg:iteration}
\textbf{Algorithm 1: Calculation of iteration step t+1} \\
\begin{enumerate}
\item 
Set the fusion variables for $j=1,...,N$
\begin{align*}
\widetilde\beta(t+1) &= \frac{1}{N}\sum_{k=1}^N \beta^k(t) - \frac{1}{N c}\sum_{k=1}^N \lambda^k(t) \\
\widetilde X_j(t+1) &= \frac{1}{q_j}\sum_{k \in \mathcal{M}_j} X_j^k(t) - \frac{1}{q_j c}\sum_{k \in \mathcal{M}_j} \Lambda_j^k(t)
\end{align*}
\item 
Calculate the new variables for $k=1,...,N$ 
\begin{align*}
&Y_k(t+1) \\
=& \text{arg}\!\!\! \min_{Y_k \in \Omega_k} \!\! \left( \!-\beta^k \!-\! \lambda^k(t) \beta^k + \frac{c}{2} | \widetilde \beta(t+1) \!-\! \beta^k|^2 \right. \\
&\left. + \!\!\!\! \sum_{j \in \mathcal{N}_k \cup k} \!\!\!\! \left( \!-tr\left(\Lambda^{k \top}_{j}\!\!(t) X_j^k \right) \! + \! \frac{c}{2}\| \widetilde X_j(t+1) \!-\! X_j^k \|^2 \right) \!\! \right)
\end{align*}
\item 
Set the Lagrange variables for $k=1,...,N$ 
\begin{align*}
\lambda^k(t+1)&=\lambda^k(t)+c(\widetilde \beta(t+1)-\beta^k(t+1))
\end{align*}
and for all $k=1,...,N, j \in \mathcal{N}_k$
\begin{align*}
\Lambda^k_j(t+1) &= \Lambda^k_j(t)+c(\widetilde X_j(t+1)-X_j^k(t+1))
\end{align*}
\end{enumerate}

\end{algorithm}

\medskip

\begin{Remark}
Out of the three steps in Algorithm 1, clearly 2) and 3) can be run in parallel by the individual estimators separately. 
Calculation of Step 1 of Algorithm 1 requires the evaluation of the mean value, which can be done in a distributed manner by applying a consensus algorithm. Under Assumption 1, average consensus algorithms can be used to calculate $\widetilde\beta(t+1)$. In particular, discrete time algorithms are preferable to keep the concept of an iterative algorithm \cite{Zhu2010} and algorithm which converge in finite-time are useful to ensure exact convergence \cite{Sundaram2007, Chen2011}. 

For the calculation of $\widetilde X_k(t+1)$ in the case of undirected graphs, only two steps are needed: All neighbors $j \in \mathcal{M}_k$ pass their $X_k^j(t)$ and $\Lambda^j_k(t)$ to estimator $k$. Then, estimator $k$ calculates $\widetilde X_k(t+1)$ and passes it back to its neighbors.
The calculation of $\widetilde X_k(t+1)$ in the case of directed graphs is more demanding with respect to the graph topology: Usual average consensus algorithms can be applied when for every $k=1,...,N$, the subgraph $\widetilde{\mathcal{G}}_k$ induced by node $k$ and its out-neighborhood $\mathcal{M}_k$, is a balanced graph. This however can be relaxed by adding additional variables $X_j^k, j \not\in \mathcal{N}_k$, to $Y_k$ and adding $X_j^k = \widetilde X_j$ as equality constraint. For instance, if for all $k=1,...,N$
$ Y_k = (F_k, \beta^k, X_1^k,...,X_N^k ) $, then $\widetilde X_k(t+1), k=1,...,N$ can be
calculated under Assumption 1 using average consensus. This will later be demonstrated in the numerical example.
\end{Remark}

\medskip

In order to show the convergence of Algorithm 1, two lemmas need to be introduced.

\medskip

\begin{Lemma}\label{Lemma:vectorization}
The Lagrangian \eqref{eq:Lagrangian_total} can be written in terms of the vectorized variables, i.e.
\begin{equation}\label{eq:Lagrangian_total_vector}
\begin{aligned}
&L(Y_k, \widetilde{\Lambda}^k) = \sum_{k=1}^N \left( -\beta^k + \lambda^k (\widetilde \beta-\beta^k) + \frac{c}{2} | \widetilde \beta-\beta^k |^2 \right) \\
 &+ \sum_{k=1}^N\sum_{j \in \mathcal{N}_k \cup k} \left( vec(\Lambda^{k}_j)^\top vec(\widetilde X_j - X_j^k) + \frac{c}{2}\| vec(\widetilde X_j - X_j^k) \|^2 \right)
\end{aligned}
\end{equation}

\begin{proof}
We have the equalities
\begin{align*}
tr(A^\top B) &=\sum_i\sum_j A_{ji}B_{ji} = vec(A)^\top vec(B)  \\
\|A\|^2 &=tr(A^\top A)=\sum_i\sum_j A_{ji}^2 = \| vec(A) \|^2 .
\end{align*}

\end{proof}
\end{Lemma}

This Lemma shows, that we can recast the problem into a problem of a standard form defined on a finite dimensional vector space.

\medskip

\begin{Lemma}\label{Lemma:boundedness_local_solution}
For fixed $\widetilde \beta, \widetilde X_k, \widetilde X_{j_1^k},..., \widetilde X_{j_{p_k}^k}, \lambda^k, \Lambda^k_{j_1^k},..., \Lambda^k_{p_k}$, the minimization
\begin{equation} \label{eq:local_Lagrangian}
\begin{aligned}
\text{arg} \min_{Y_k \in \Omega_k} & \left( -\beta^k - \lambda^k \beta^k + \frac{c}{2} | \widetilde \beta-\beta^k |^2 \right. \\
& \left. + \sum_{j \in \mathcal{N}_k \cup k}  \left( -tr\left(\Lambda^{k \top}_{j} X_j^k \right) + \frac{c}{2}\| \widetilde X_j - X_j^k \|^2 \right)\right)
\end{aligned}
\end{equation}
is always attainable.
\end{Lemma}

\begin{proof} 
First, note that the LMI conditions \eqref{LMI:stability_sep}-\eqref{LMI:rest_feasible_set_sep} are non-strict inequalities. The definition range of the solution matrices $X_k^k >0, X_j^k>0, j\in\mathcal{N}_k$ are strict inequalities, but \eqref{LMI:stability_sep}-\eqref{LMI:rest_feasible_set_sep} imply that there exists a $\overline{\delta}>0$ such that $X_k^k \geq \overline{\delta} \; I$ and $X_j^k \geq \overline{\delta} \; I$ for $j \in \mathcal{N}_k$. Thus, the feasible set $\Omega_k$ is closed and convex.


Following again the proof of Theorem 1, \eqref{LMI:stability_sep} and \eqref{LMI:rest_feasible_set_sep} imply that $F_k, \beta^k, X_k^k$ are bounded for all $k=1,...,N$. In constrast, the variables $X_j^k$ for $ j\in \mathcal{N}_k$ are not restricted to a bounded set by the LMIs \eqref{LMI:stability_sep}-\eqref{LMI:rest_feasible_set_sep}. 
However, note that the cost function of \eqref{eq:local_Lagrangian} is quadratic in the variables $X^k_j, j \in \mathcal{N}_k$. Thus, due to the boundedness of $F_k, \beta^k, X_k^k$, we conclude that the sub-level sets of \eqref{eq:local_Lagrangian}
\begin{equation}\label{eq:sublevelset}
\begin{aligned}
& \left\{ \left. Y_k \in \Omega_k | \;  -\beta^k - \lambda^k \beta^k + \frac{c}{2} | \widetilde \beta-\beta^k |^2 \right. \right. \\
& \left. + \sum_{j \in \mathcal{N}_k \cup k}  \left( -tr\left(\Lambda^{k \top}_{j} X_j^k \right) + \frac{c}{2}\| \widetilde X_j - X_j^k \|^2 \right) < \overline{c} \right\}
\end{aligned}
\end{equation}
for $\overline{c} \in \mathbb{R}$ are bounded.
Following the argument in Proposition 4.1 in \cite{Bertsekas1989}, Chapter 3, we can conclude that we can equivalently search for the minimum of the cost function over a non-empty sub-level set \eqref{eq:sublevelset} instead of $\Omega_k$. 
Therefore, we can conclude that \eqref{eq:local_Lagrangian} is always attainable.

 
\end{proof}

\medskip

\begin{Theorem}
Algorithm~1 is a solution to Problem 1. In particular, the iteration steps $Y_k(t), k=1,...,N,$ can be calculated in parallel, and satisfy the convergence conditions \eqref{eq:beta-convergence}, \eqref{eq:matrix-convergence}.
\end{Theorem}

\begin{proof}
Using Lemma \ref{Lemma:vectorization} and \ref{Lemma:boundedness_local_solution}, we can follow the steps from \cite{Bertsekas1989}, Section 3.3 and 3.4, in order to prove convergence of the iterations. 

\end{proof}

\section{Numerical Example}
Like in \cite{Ugrinovskii2011a}, we consider a system of the form \eqref{eq:LTI}, with
\begin{equation*}
\begin{aligned}
A &= \begin{bmatrix}
0.3775 & 0& 0& 0& 0& 0 \\
0.2959 & 0.3510& 0& 0& 0& 0 \\
1.4751 & 0.6232& 1.0078& 0& 0& 0 \\
0.2340 & 0& 0& 0.5596& 0& 0 \\
0& 0& 0& 0.4437& 1.1878& −0.0215 \\
0& 0& 0& 0& 2.2023& 1.0039
\end{bmatrix}, \\
B &= \begin{bmatrix}
0.1 \; I_6 & 0
\end{bmatrix} \quad\quad \overline{D}_k = 0.01 \; I_2  \text{ for all } k=1,...,N
\end{aligned}
\end{equation*}
which is observed by six sensor nodes, sensing two coordinates each. For every sensor an estimator is implemented, where none of the estimators is able to estimate the complete state vector without communication. The communication topology is assumed to be a directed circulant graph and we use Algorithm 1 to calculate the filter gains. For the numerical calculations we use YALMIP \cite{Lofberg2004}. Since we are dealing with a directed but balanced graph, we apply the method described in Remark 4 and use complete local representations of all variables $X_j$ at every estimator $k$.
The algorithm is run with both fixed performance parameter $\beta^{const}=100$ as discussed in Remark 3, and also using optimization over the variable performance parameter $\beta^k(t), k=1,...,N,$ \eqref{prob:separable_optimization}, where $\beta^k(0)=100$ is set as  initialization.

In the first case, where $\beta$ is fixed, we evaluate the matrix convergence condition \eqref{eq:matrix-convergence} by calculating  the average value $X_j^{ave}=\frac{1}{q_j}\sum_{k \in \mathcal{M}_j \cup j} X^k_j  $
and subsequently
$\text{Error} = \sum_{j=1}^N \sum_{k = 1}^N \| X^k_j - X_j^{ave}\|^2$.

In the second case, involving optimisation over $\beta_k$, we additionally calculate $\beta^{ave}=\frac{1}{N}\sum_{k=1}^N \beta^k$
and subsequently we have
$\text{Error} = \sum_{j=1}^N \sum_{k = 1}^N \| X^k_j - X_j^{ave}\|^2 + \sum_{k = 1}^N | \beta^k - \beta^{ave}|^2$.

The plots of the error evolution are shown in Figure \ref{fig:fixed_beta} and \ref{fig:variable_beta}. Figure
\ref{fig:variable_beta} additionally shows the evolution of $\beta^{ave}$. The graph demonstrates that $\beta^{ave}$ is monotonically  increasing, and since it is bounded from above according to Theorem 1, it must eventually converge to a limit. In fact, it eventually converges to $2.3 \cdot 10^3$.


 \begin{figure}
  \centering
%
%
%
%
\begin{tikzpicture}[scale=0.9]

\begin{axis}[%
width=3in,
height=1.5in,
scale only axis,
xmin=0.5,
xmax=70.5,
xlabel={Iteration},
ymin=0,
ymax=50,
ylabel={Error}
]
\addplot [
color=blue,
only marks,
mark=*,
mark options={solid},
forget plot
]
table[row sep=crcr]{
1 47.1675204874294\\
2 12.2219068595352\\
3 7.76659969357437\\
4 13.1610943746415\\
5 13.0386886092203\\
6 8.91019171082914\\
7 6.12844650108794\\
8 4.30788782049918\\
9 3.8837520932762\\
10 3.57719862848702\\
11 3.16473255624788\\
12 2.90969085651121\\
13 2.76195632230267\\
14 2.54837148109725\\
15 2.27542333444472\\
16 1.97178364417098\\
17 1.70080663386874\\
18 1.5465340403081\\
19 1.46162327357446\\
20 1.40502013520312\\
21 1.34729169340615\\
22 1.28179379076706\\
23 1.21768111937641\\
24 1.16819933981131\\
25 1.14212578878581\\
26 1.14042174239792\\
27 1.15760946063356\\
28 1.18592968431011\\
29 1.2195342772008\\
30 1.25630922906328\\
31 1.24616054570133\\
32 1.23888919764846\\
33 1.21387610117458\\
34 1.17128909186928\\
35 1.11745264439127\\
36 1.06119066714838\\
37 1.01052421071453\\
38 0.969681606900516\\
39 0.938066892980075\\
40 0.911291079580615\\
41 0.883429817387642\\
42 0.849414562685127\\
43 0.807214301288349\\
44 0.757835703510637\\
45 0.703796256864894\\
46 0.652467687150293\\
47 0.553057854911856\\
48 0.458319443385305\\
49 0.416628260331813\\
50 0.406997415769057\\
51 0.419588747149699\\
52 0.440941740903935\\
53 0.35023768409037\\
54 0.345529346947959\\
55 0.346942164815517\\
56 0.341471581627327\\
57 0.333069279977993\\
58 0.325077733018393\\
59 0.319364152681972\\
60 0.316565768449267\\
61 0.316637144449978\\
62 0.319282525214597\\
63 0.32405508779634\\
64 0.330226162070331\\
65 0.336328364049205\\
66 0.268265511778113\\
67 0.0390196309037569\\
68 0.0421747740860581\\
69 0.0465890645580112\\
70 0.0530596362040954\\
};
\end{axis}
\end{tikzpicture}%
 	\caption{Evolution of the error during iteration for fixed $\beta$.}	
 	\label{fig:fixed_beta}
 \end{figure}
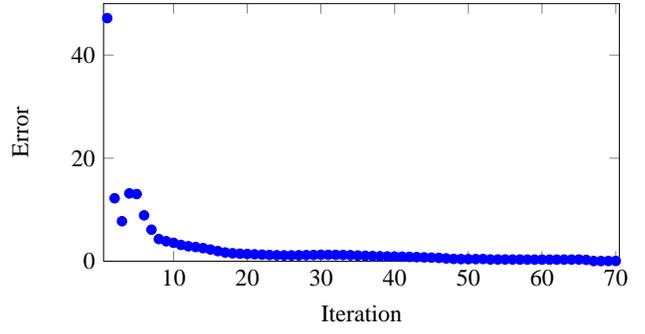
 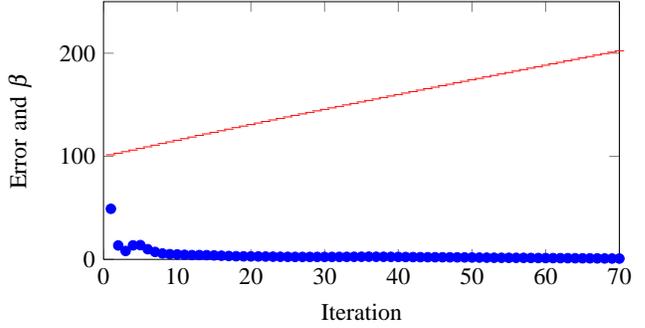
\begin{figure}
  \centering
%
%
%
%
\begin{tikzpicture}[scale=0.9]

\begin{axis}[%
width=3in,
height=1.5in,
scale only axis,
xmin=0,
xmax=70,
xlabel={Iteration},
ymin=0,
ymax=250,
ylabel={Error and $\beta$}
]
\addplot [
color=red,
only marks,
mark=-,
mark options={solid},
forget plot
]
table[row sep=crcr]{
1 101.498285126458\\
2 102.994831458464\\
3 104.494457974958\\
4 106.030181434215\\
5 107.598939224142\\
6 109.189135053169\\
7 110.790477132065\\
8 112.38732023638\\
9 113.970652411284\\
10 115.541585508216\\
11 117.102069159339\\
12 118.653543874904\\
13 120.197454080086\\
14 121.727748377579\\
15 123.248613564189\\
16 124.766217082627\\
17 126.278298889035\\
18 127.781830327618\\
19 129.280542553257\\
20 130.774058740066\\
21 132.262369685722\\
22 133.746125462127\\
23 135.225794445256\\
24 136.701617463198\\
25 138.17366228491\\
26 139.641874799461\\
27 141.106195831871\\
28 142.566631496986\\
29 144.023267881998\\
30 145.476374485755\\
31 146.926350858276\\
32 148.373676589788\\
33 149.818845805035\\
34 151.262335018497\\
35 152.704559285093\\
36 154.145740769821\\
37 155.586237019862\\
38 157.023839612479\\
39 158.459637798482\\
40 159.893745258397\\
41 161.326210975739\\
42 162.757123476233\\
43 164.186600116671\\
44 165.614781786027\\
45 167.041819270105\\
46 168.467839453811\\
47 169.892943868739\\
48 171.317186509078\\
49 172.740589379994\\
50 174.163135417902\\
51 175.584787095442\\
52 177.005484193831\\
53 178.425164059621\\
54 179.843768209167\\
55 181.261251778595\\
56 182.677595165243\\
57 184.092807976996\\
58 185.506810825991\\
59 186.918350556023\\
60 188.328360975602\\
61 189.737136153103\\
62 191.144919274771\\
63 192.551904119589\\
64 193.958207243919\\
65 195.363881229224\\
66 196.768735951032\\
67 198.172848257959\\
68 199.576226058898\\
69 200.978849505579\\
70 202.38054948703\\
};
\addplot [
color=blue,
only marks,
mark=*,
mark options={solid},
forget plot
]
table[row sep=crcr]{
1 49.0493572962513\\
2 13.487858868637\\
3 8.1326562913952\\
4 13.5779881593236\\
5 13.8448102548173\\
6 9.8938350398295\\
7 7.21966913008623\\
8 5.70794480489304\\
9 5.06748511857303\\
10 4.81752289009432\\
11 4.41485903676168\\
12 4.16392494987598\\
13 4.18191126336274\\
14 4.09896246795498\\
15 3.83983768251121\\
16 3.54406223239355\\
17 3.28589334350545\\
18 3.05973639833885\\
19 2.9534060717922\\
20 2.88835609022362\\
21 2.82150198430992\\
22 2.73503779850036\\
23 2.63324581789027\\
24 2.53323283847757\\
25 2.45336822862591\\
26 2.40237734962827\\
27 2.37809852202292\\
28 2.37179926597604\\
29 2.37360105743285\\
30 2.37878964189129\\
31 2.38870472515732\\
32 2.40760725860445\\
33 2.4388703321886\\
34 2.48287835958479\\
35 2.53740637413021\\
36 2.59966919085595\\
37 2.52454487825028\\
38 2.49875003423771\\
39 2.44324189609305\\
40 2.36742487734796\\
41 2.28506027979305\\
42 2.2098405046912\\
43 2.15034943562857\\
44 2.1079854472157\\
45 2.07699388883086\\
46 2.04756930355857\\
47 2.00999926432161\\
48 1.95817825250753\\
49 1.89123392986594\\
50 1.81329151147467\\
51 1.73180413886133\\
52 1.65529877596593\\
53 1.5910294850755\\
54 1.54320780001918\\
55 1.51222650197254\\
56 1.49497985637435\\
57 1.47248524454225\\
58 1.3774894370354\\
59 1.3210420724047\\
60 1.26737880922475\\
61 1.22344561139019\\
62 1.18631157990523\\
63 1.14514445486218\\
64 1.1016558459055\\
65 1.0517425660733\\
66 1.00367051530068\\
67 0.958117966908729\\
68 0.891413424350626\\
69 0.825262598947095\\
70 0.788123040693309\\
};
\end{axis}
\end{tikzpicture}%
 	\caption{Evolution of the error (blue dots) and performance index $\beta^{ave}$ (red line) in the algorithm involving optimisation over variables $\beta^k$.}	
 	\label{fig:variable_beta}
 \end{figure}
 
Better performance $\beta$ however is achieved at the expense of higher filter gains. For instance, after $70$ iterations, the consensus gain $K_1$ is 
\begin{align*}
\begin{bmatrix}
   21.1005  & -0.0256  &  0.0196  & -0.6018  &  0.0418 &   0.0117 \\
   -0.0215  & 73.3369  &  0.5599  &  0.0073  &  0.0025 &   0.0021 \\
   -0.0423  & -0.8806  & 99.8791  &  0.0617  &  0.0981 &   0.0536 \\
   -0.6033  & -0.0178  &  0.0618  & 70.6692  &  1.3701 &   2.8005 \\
    0.0415  & -0.0054  &  0.0972  &  1.7726  & 20.7775 &   5.0466 \\
    0.0117  & -0.0003  &  0.0554  &  2.7740  &  3.2740 &  17.7281 
\end{bmatrix}
\end{align*}
in the fixed-$\beta$ case and
\begin{align*}
\begin{bmatrix}
 28.8328  & -0.0572  &  0.0291  & -0.0089 &   0.0333  &  0.0655 \\
   -0.0397 &  99.9887  &  0.8962   & 0.0198  &  0.0050  &  0.0044 \\
   -0.0222 &  -0.8149 & 100.0003 &   0.2984 &   0.1414  &  0.0747 \\
   -1.7921  & -0.0044 &  0.3121 &   71.0412  &  1.4502   & 2.5260 \\
    0.0751&   -0.0035  &  0.1384  &  1.8242  & 27.2461 &   6.2623 \\
   -0.0177  &  0.0016  &  0.0711  &  2.4613  &  4.4543 &  23.0606
\end{bmatrix}
\end{align*}
in the variable-$\beta$ case.


\section{Conclusion}

We have developed a method for distributed filter design for cooperative $\mathcal{H}_\infty$-type estimation. In order to achieve this we separated the centralized problem by introducing additional variables and then applied an algorithm that works locally and only needs communication for average consensus.


\bibliographystyle{unsrt}
\bibliography{bibliography}

\end{document}